\documentclass[sigconf]{aamas}
\pdfoutput=1
\newcommand{\extended}[1]{}    
\newcommand{\short}[1]{#1}     

\usepackage[T1]{fontenc}
\usepackage{xspace}
\usepackage{enumitem}
\usepackage{graphics,graphicx}
\usepackage{epsf}\usepackage{longtable}
\usepackage{latexsym}
\usepackage{amsfonts}
\usepackage{color}
\usepackage[mathscr]{eucal}
\usepackage{xspace}
\usepackage{tabularx}
\usepackage{mathtools}
\usepackage{tikz}
\usetikzlibrary{arrows,automata,calc,shapes,decorations,backgrounds,petri,mindmap,fit,positioning}
\usepackage{balance}

\newcommand{\lan}[1]{\ensuremath{\mathbf{#1}}\xspace}

\newcommand{\stratstyle}[1]{\ensuremath{\mathrm{#1}}}

\newcommand{\iR}{\stratstyle{iR}\xspace}

\newcommand{\ATL}[1][]{\lan{ATL_{\stratstyle{#1}}}}

\newcommand{\SL}{\lan{SL}}

\newcommand{\CSL}[1][]{\lan{CSL}}
\newcommand{\CSLP}[1][]{\lan{CSLP}}

\newcommand{\MIATL}[1][]{IATL[R]}

\newcommand{\ATLES}[1][]{\lan{ATLES}}
\newcommand{\ATELA}[1][]{\lan{ATELA}}

\newcommand{\LTL}{\lan{LTL}}

\newcommand{\coop}[2][]{\langle\!\langle{#2}\rangle\!\rangle_{_{\!\mathit{#1}}}}

\newcommand{\Next}[1][]{\!\raisebox{-.2ex}{ \mbox{\unitlength=0.9ex
            \begin{picture}(2,2)
            \linethickness{0.06ex}
            \put(1,1){\circle{2}}   \end{picture}}}_{{#1}}  \,}
\newcommand{\Sometm}[1][]{\Diamond_{{#1}}}
\newcommand{\Always}[1][]{\Box_{{#1}}}

\renewcommand{\Next}{\mathrm{X}\,}
\renewcommand{\Sometm}{\mathrm{F}\,}
\renewcommand{\Always}{\mathrm{G}\,}
\newcommand{\Until}{\,\mathrm{U}\,}

\newcommand{\Release}{\,\mathrm{R}\,}

\newcommand{\Agt}{{\ensuremath{\mathbb{A}\mathrm{gt}}}}
\newcommand{\States}{St}

\newcommand{\Props}{\ensuremath{PV}\xspace}
\newcommand{\V}{V}

\newcommand{\trans}{o}

\newcommand{\TransArrow}[1][]{\hookrightarrow....}

\newcommand{\M}{M}
\newcommand{\model}{\M}

                                    \newcommand{\onepath}[1][]{\ensuremath{\lambda\ifthenelse{\equal{#1}{}}{}{[#1]}}}         

\newcommand{\verf}{\mathbf{v}}
\newcommand{\reft}{\mathbf{r}}

\newcommand{\str}{s}

\newcommand{\prop}[1]{\ensuremath{\mathsf{{#1}}}}

\newcommand{\plaus}[1][]{\ifthenelse{\equal{#1}{}}{\mathbf{P\;\!\!l}\,}{\mathbf{P\;\!\!l}_{#1}\,}}
\newcommand{\phys}[1][]{\ifthenelse{\equal{#1}{}}{\mathbf{P\;\!\!h}\,}{\mathbf{P\;\!\!h}_{#1}\,}}

\newcommand{\plaumodels}[1][]{\ensuremath{\ifthenelse{\equal{#1}{}}{\models_\sPlaupaths}{\models_{#1}}}}
\newcommand{\sPlaupaths}{\ensuremath{P}}

\newcommand{\then}{\rightarrow}

\newcommand{\true}{\top}
\newcommand{\false}{\bot}

\newcommand{\satisf}[1][]{\models_{_{#1}}}

\usepackage{color}

\definecolor{lightgrey}{rgb}{0.8,0.8,0.8}
\definecolor{grey}{rgb}{0.6,0.6,0.6}
\definecolor{darkgrey}{rgb}{0.4,0.4,0.4}
\definecolor{darkgreen}{rgb}{0,0.7,0}

\newcommand{\onlabel}[1]{\colorbox{white}{\textit{{#1}}}}

\newcommand{\set}[1]{\{{#1}\}}
\newcommand{\powerset}[1]{2^{#1}}
\newcommand{\tuple}[1]{\langle{#1}\rangle}

\newcommand{\Nat}{\mathbb{N}}

\newcommand{\complexityclass}[1]{\ensuremath{\mathbf{{#1}}}\xspace}

\newcommand{\Ptime}{\complexityclass{P}}
\newcommand{\PTIME}{\complexityclass{P}}
\newcommand{\NP}{\complexityclass{NP}}

\newcommand{\Exptime}{\complexityclass{EXPTIME}}
\newcommand{\EXPTIME}{\Exptime}

\newcommand{\LOGTIME}{\complexityclass{DLOGTIME}}

\newcommand{\putaway}[1]{}

\newcommand{\para}[1]{\smallskip\noindent\textbf{#1}}

\newenvironment{itemize2}{\begin{itemize}\itemsep 0in}{\end{itemize}}
\newenvironment{enumerate2}{\begin{enumerate}\itemsep 0in}{\end{enumerate}}
\newenvironment{description2}{\begin{description}\itemsep 0in}{\end{description}}

\newcommand{\finis}{{\scriptsize $\blacksquare$}}
\newcommand{\finisdef}{$\Box$}
\newcommand{\bul}{{\tiny $\blacksquare$}}

\def\itemiremember{\labelitemi}
\def\itemiiremember{\labelitemii}

\newtheorem{remark}{Remark}

\definecolor{tucgreen}{RGB}{0,140,79}

\renewcommand{\str}{s}

\newcommand{\WJ}[1]{{\color{tucgreen}WJ: #1}\xspace}
\renewcommand{\WJ}[1]{}
\newcommand{\CD}[1]{{\color{blue}CD: #1}\xspace}
\newcommand{\todo}[1]{{\color{red}TO DO: #1}\xspace}

\newcommand{\complx}{\mathcal{C}}
\newcommand{\Keys}{\mathcal{K}}
\newcommand{\enc}{\mathit{enc}}
\newcommand{\timecompl}[1]{\mathit{time}_{\,#1}}
\newcommand{\run}{\rho}
\newcommand{\inseq}{\alpha}
\newcommand{\strat}{\str}
\newcommand{\strset}[2][]{\mathbf{str}_{#1,#2}}
\newcommand{\genstrset}[2][]{\mathbf{Str}_{#1,#2}}
\newcommand{\genstr}{\pmb{\mathit{S}}}
\newcommand{\strtemp}{\genstr}
\newcommand{\mclass}[1][]{\mathcal{\model}_\mathit{#1}}
\newcommand{\stratunif}[1][]{{\;|\!\!\!\equiv_{#1}\;}}
\newcommand{\notstratunif}[1][]{{\;|\!\!\!\not\equiv_{#1}\;}}
\newcommand{\stratadapt}[1][]{{\;|\!\!\!=_{#1}\;}}

\newcommand{\initstate}{\state_0}
\newcommand{\repert}{R}
\renewcommand{\state}{q}
\renewcommand{\trans}{t}
\newcommand{\coffee}{\mathit{coffee}}

\setcopyright{ifaamas}
\acmConference[AAMAS '23]{Proc.\@ of the 22nd International Conference
on Autonomous Agents and Multiagent Systems (AAMAS 2023)}{May 29 -- June 2, 2023}
{London, United Kingdom}{A.~Ricci, W.~Yeoh, N.~Agmon, B.~An (eds.)}
\copyrightyear{2023}
\acmYear{2023}
\acmDOI{}
\acmPrice{}
\acmISBN{}

\acmSubmissionID{1237}

\title{Computationally Feasible Strategies}

\author{Catalin Dima}
\affiliation{
  \institution{LACL, Universit\'e Paris Est Cr\'{e}teil}
\city{F-94010 Cr\'{e}teil}
  \country{France}}
\email{dima@u-pec.fr}

\author{Wojciech Jamroga}
\affiliation{
  \institution{Institute of Computer Science, Polish Academy of Sciences}
  \city{and SnT, University of Luxembourg}
  \country{}}
\email{wojciech.jamroga@uni.lu}

\begin{abstract}
Real-life agents seldom have unlimited reasoning power. In this paper, we propose and study a new formal notion of computationally bounded strategic ability in multi-agent systems. The notion characterizes the ability of a set of agents to synthesize an executable strategy in the form of a Turing machine within a given complexity class, that ensures the satisfaction of a temporal objective in a parameterized game arena. We show that the new concept induces a proper hierarchy of strategic abilities -- in particular, polynomial-time abilities are strictly weaker than the exponential-time ones. We also propose an ``adaptive'' variant of computational ability which allows for different strategies for each parameter value, and show that the two notions do not coincide. Finally, we define and study the model-checking problem for computational strategies. We show that the problem is undecidable even for severely restricted inputs, and present our first steps towards decidable fragments.
\end{abstract}


\begin{document}

\pagestyle{fancy}
\fancyhead{}

\maketitle

\section{Introduction}\label{sec:intro}

Multi-agent systems (MAS) involve the interaction of multiple autonomous agents, often assumed to exhibit self-interested, goal-directed behavior~\cite{\extended{Weiss99mas,Wooldridge02intromas,}Shoham09MAS}.
Many relevant properties of MAS can be phrased in terms of \emph{strategic abilities} of agents and their groups~\cite{Goranko15stratmas,Agotnes15handbook}. In particular, most functionality requirements can be specified as the ability of the authorized agents to achieve their legitimate goals, or to complete their tasks.
Moreover, many security properties refer to the inability of the ``bad guys'' to obtain their goals.
For example, a sensible requirement for a coffee vending machine is that a thirsty user can eventually get a cup of the aromatic brew if she follows the right sequence of steps.
Similarly, one might want to require that a secure file system disables impersonation, i.e., the intruder cannot log in as another user, no matter how smart strategy he chooses to play.

The ``right'' semantics of strategic ability in MAS has been a matter of hot debate since early 2000s, cf.~e.g.~\cite{Agotnes04atel,Jamroga04ATEL,Schobbens04ATL,Dima10communicating,Guelev12stratcontexts,Mogavero14behavioral,Aminof16promptATL}.
In particular, it was argued that ability-based requirements should take into account the bounded capabilities of agents~\cite{Jamroga19natstrat-aij,Jamroga19natstratii}.
Typically, strategies are defined as functions from sequences of system states (i.e., possible histories of the play) to the agent's choices.
The approach is mathematically elegant, and might be appropriate to reason about abilities of agents with unlimited time and reasoning power.
However, it makes for a poor model of human behavior, whose strategies should be \emph{easy to obtain}, \emph{easy to remember}, and \emph{efficient to use}~\cite{Nielsen94usability}; in other words, simple enough to be useful.
Similarly, one often wants to look only at threats from \emph{computationally bounded attackers}, for example in cases when {unconditional security} is unattainable~\cite{Diffie76crypto-newdirections,KatzLindell20crypto}.

In this paper, we propose and formalize the concept of \emph{computationally bounded strategic ability}.
We draw from two main sources of inspiration.
First, it was proposed in~\cite{Jamroga19natstrat-aij,Jamroga19natstratii} to formalize human ability with statements $\coop{A}^{\le k}\varphi$, saying that agents $A$ have a strategy of \emph{complexity at most $k$} to achieve goal $\varphi$. For instance, one may require that the user has a strategy of complexity at most $10$ to get an espresso (think of a sheet of paper with a recipe that uses only $10$ symbols or less).
Unfortunately, such concrete bounds are usually arbitrary. Why is a recipe of length $10$ still good for buying a coffee, but one with $11$ symbols is too complex? Why not up to $15$ symbols? A possible solution is offered by asymptotic bounds of complexity theory. For a scalable model of the coffee machine, parameterized by the increasing sophistication of the brew it produces, we may demand that the complexity of getting an espresso grows \emph{at most linearly} with the sophistication of the machine.\footnote{
  At most quadratically for Neapolitan users. } Note that the goal is the same regardless of the sophistication of the machine: getting an espresso.

Secondly, our approach is inspired by cryptographic definitions of security, where the space of potential attack strategies is defined by polynomial-time probabilistic Turing machines~\cite{Bellare98crypto-notions,KatzLindell20crypto}. However, with the advances in computing techniques (in particular quantum computers), it is unclear how long polynomial time will remain the boundary separating feasible and unfeasible attacks.

To address this, we propose a general framework that allows us to fix an arbitrary complexity class as the boundary, and ask if there exists a winning Turing-representable strategy within the class, which
achieves a fixed objective on all the game structures from an effective class.
We consider strategies represented by deterministic Turing machines; the probabilistic case is left for future work.
In terms of technical results, we show that different complexity classes induce different strategic abilities, and thus the concept of computationally-bounded strategies gives rise to a new hierarchy of abilities.
Moreover, we prove that uniform abilities (i.e., ones where a single general strategy solves all the game instances) do \emph{not} coincide with adaptive abilities (i.e., ones where different strategies can be used for different instances).
Finally, we define the model checking problem and show that it is inherently undecidable. We conclude by indicating a couple of decidable cases that can serve as a starting point to obtain more general characterizations in the future.

\para{Related work.}
Logical formalizations of strategic ability are usually based on alternating-time logic \ATL~\cite{Alur02ATL} or strategy logic \SL~\cite{Chatterjee07strategylogic,Mogavero10stratLogic}.
This includes works discussing different representations and restrictions on strategies, e.g.,~\cite{Jamroga04ATEL,Schobbens04ATL,Dima10communicating,Guelev11atl-distrknowldge,Vester13ATL-finite,Mogavero14behavioral}, and establishing the complexity of model checking for various strategy types, cf.~\cite{Bulling10verification} for an overview.
The models and decidability results for bounded- and finite-memory strategies are of particular relevance~\cite{Vester13ATL-finite}.
However, none of the frameworks takes into account the complexity of the strategies that the agents need to achieve their objective in the game -- with the notable exception of \cite{goranko-kuusisto}
where the complexity of the strategy is considered with respect to the formula size, not the model size, whereas in our case the same objective must be achieved in all models.
Complexity bounds have been only used in the natural strategy-based extensions of \ATL~\cite{Jamroga19natstrat-aij,Jamroga19natstratii} and \SL~\cite{Belardinelli22auctions}, and that just in a rudimentary form (concrete rather than asymptotic).

Our research bears a conceptual connection to \emph{bounded rationality} in game theory, where the players have limited computational power when reasoning about their choices~\cite{Rubinstein98bounded,Shoham09MAS}, e.g., by imposing bounds on their memory~\cite{Kocherlakota98moneismemory,Horner09folktheorem,Bhaskar12socialmemory}.
Also, our distinction between uniform and adaptive ability is related to \emph{uniform vs.~non-uniform complexity}~\cite{Cai12complexity,Shannon49circuits,Furst81circuits}. Specifically, uniform and non-uniform notions of security were discussed in~\cite{Goldreich93nonuniform,De09nonuniform,Bernstein13nonuniform,Koblitz13nonuniformity}.
Parameterized model checking~\cite{Aminof16parameterisedGrid} is also related.
The above connections are rather loose; in particular, we define uniformity w.r.t. Turing machines and not Boolean circuits like in the standard approach.
 \section{Preliminaries}\label{sec:preliminaries}

\para{Models of interaction.}
We consider multi-step games with imperfect information, played by multiple agents over a finite game arena with concurrent synchronous moves~\cite{Alur02ATL,Hoek02ATEL,Schobbens04ATL}.

\begin{definition}[Model]\label{def:cgs}
An \emph{imperfect information concurrent game structure (iCGS)}, or simply a \emph{model}, is given by a tuple \linebreak
$\model = \tuple{\Agt, \States, \initstate, \Props, \V, Act, \repert, \trans,\set{\sim_a}_{a\in \Agt}}$
which includes a non\-empty finite set of all agents $\Agt = \set{a_1,\dots,a_k}$, a non\-empty finite set of states $\States$, an initial state $\initstate\in\States$, a set of atomic propositions $\Props$ and their valuation $\V : \Props\rightarrow \powerset{\States}$, and a nonempty finite set of (atomic) actions $Act$. The repertoire function $\repert : \Agt \times \States \rightarrow \powerset{Act}\setminus\{\emptyset\}$ defines nonempty sets of actions available to agents at each state; we will write $\repert_a(\state)$ instead of $\repert(a,\state)$, and define $\repert_A(\state) = \prod_{a\in A}\repert_a(\state)$ for each $A\subseteq\Agt, \state\in\States$.
Furthermore, $\trans$ is a (deterministic) transition function that assigns the outcome state $\state' = \trans(\state,\alpha_1,\dots,\alpha_k)$ to each state $\state$ and tuple of actions $\tuple{\alpha_1, \dots, \alpha_k}$ such that $\alpha_i \in \repert(a_i,\state)$ for $i=1, \dots, k$.

Every $\sim_a\subseteq \States\times\States$ is an \extended{epistemic }equivalence relation with the intended meaning that, whenever $q\sim_a q'$, the states $q$ and $q'$ are indistinguishable to agent $a$.
The iCGS is assumed to be \emph{uniform}, in the sense that $q\sim_a q'$ implies $\repert_a(q)=\repert_a(q')$, i.e., the same choices are available in indistinguishable states.
Note that perfect information can be modeled by assuming each $\sim_a$ to be the identity relation.
\end{definition}

The \emph{observations} of agent $a$ are defined as $Obs_a = \set{[\state]_{\sim_a} \mid \state\in\States}$, i.e., the equivalence classes of the indistinguishability relation.
Assuming an arbitrary ordering of states (e.g., the shortlex ordering, also known as the length-lexicographic ordering), we can uniquely identify each observation $o\in Obs$ with the minimal element of $o$.
Note that the notion of repertoire lifts to observations in a straightforward way.

\para{Strategies and their outcomes.}
A \emph{strategy} of agent $a\in\Agt$ is a conditional plan that specifies what $a$ is going to do in every possible situation.
The most general variant is given by so called \emph{\iR strategies}, i.e., strategies with imperfect information and perfect recall~\cite{Schobbens04ATL}.
An \iR strategy for $a$ can be formally represented by a function $\strat_a : {Obs_a}^+\to Act$ such that $s_a(o_0\dots o_n)\in\repert_a(o_n)$.

A collective strategy $\strat_A$ for coalition $A\subseteq\Agt$ is simply a tuple of individual \iR strategies, one per agent in $A$.

A \emph{path} $\lambda=\state_0\state_1\state_2\dots$ is an infinite sequence of states such that there is a transition between each $\state_i,\state_{i+1}$.
We use $\lambda[i]$ to denote the $i$th position on path $\lambda$ (starting from $i=0$) and $\lambda[i,j]$ to denote the part of $\lambda$ between positions $i$ and $j$, which includes the case $j = \infty$.
Function $out(\model,\strat_A)$ returns the set of all paths that can result from the execution of strategy $\strat_A$ in model $\model$.
For agents outside $A$, path transitions can involve any actions allowed by their repertoires.

\para{\LTL objectives.}
To specify the ``winning condition'', i.e. the objective that the coalition wants to achieve, we will use the formulas of \emph{linear time logic \LTL}~\cite{Pnueli77temporal,Emerson90temporal}, built on atomic propositions $\Props$ that occur in the models. The syntax of \LTL is given by the following grammar:
\begin{center}
$\varphi ::= \prop{p} \mid \neg\varphi \mid \varphi\land\varphi \mid \Next\varphi \mid \varphi\Until\varphi$,
\end{center}
where $\prop{p}\in\Props$\extended{ represents atomic propositions}, $\Next$ stands for ``in the next moment,'' and $\Until$ is the temporal operator known as ``strong until.''
The Boolean constants $\false,\true$ and the other Boolean operators are defined in the standard way.
Moreover, the temporal operator ``release'' can be defined as $\varphi_1\Release\varphi_2 \equiv \neg((\neg\varphi_1)\Until(\neg\varphi_2))$,
``sometime in the future'' as $\Sometm\varphi \equiv \true\Until\varphi$, and ``always from now on'' as $\Always\varphi \equiv \false\Release\varphi$.

The semantics of \LTL is given by the following clauses:
\begin{description2}
\item[{$\lambda \satisf[\LTL] \prop{p}$}] iff $\lambda[0] \in \V(\prop{p})$;
\item[{$\lambda \satisf[\LTL] \neg\varphi$}] iff $\lambda \not\satisf[\LTL] \varphi$;
\item[{$\lambda \satisf[\LTL] \varphi_1\land\varphi_2$}] iff  $\lambda\satisf[\LTL] \varphi_1$
  and $\lambda \satisf[\LTL] \varphi_2$;
\item[{$\lambda \satisf[\LTL] \Next\varphi$}] iff $\lambda[1,\infty] \satisf[\LTL] \varphi$;
\item[{$\lambda \satisf[\LTL] \varphi_1\Until\varphi_2$}] iff
  $\lambda[i,\infty]\ \satisf[\LTL]\ \varphi_2$ for some $i\ge 0$ and $\lambda[j,\infty] \satisf[\LTL] \varphi_1$
  for all $0\leq j< i$.
\end{description2}

We say that strategy $\strat_A$ enforces objective $\varphi$ in model $\model$ iff $\lambda \satisf[\LTL] \varphi$ for every $\lambda\in out(\model,\strat_A)$.
 \section{Computational Strategic Ability}\label{sec:compgames}

In this section, we introduce the concept of computational strategies, and show that different complexity bounds produce different views of agents' ability.

\subsection{Model Templates}\label{sec:models}

We begin by the introduction of scalable models, in which one can look for a general strategy that wins for all the values of the scalability parameter.

\begin{definition}[Model template]\label{def:modeltemplate}
A \emph{model template} is a countable a family of models $\mclass = (\model_1,\model_2,\dots)$.
Each model $\model\in\mclass$ is a finite iCGS with the same set of all agents $\Agt$, actions drawn from the same set of actions $Act$, and atomic propositions drawn from the same set $\Props$.
\end{definition}

Typically, $\mclass$ will be generated by some ``security parameter'' $k$ with a countable domain of values $\Keys$.
Without loss of generality, we usually assume that $\Keys=\Nat$.
Operationally, this can be understood as having a procedure that, given a value of $k$, instantaneously returns the model $\model = \mclass(k)$.

\begin{example}[Fibonacci espresso machine]\label{ex:fibonacci}
Consider a scalable coffee machine that simultaneously provides $n\ge 2$ cups of espresso, numbered $i=1,2,\dots,n$, to $n$ users. 
The machine has two buttons for controlling the addition of sugar in the espresso cups, with the numbers and the order of the cups that receive sugar being computed as follows:
Two designated users (agents), Alice and Bob, need to press one of the two buttons, with  Alice doing the first $\lceil\frac{n}{2}\rceil$ choices and Bob the remaining ones. 
The machine counts the number $m\le n$ of requests for sugar, then computes the $m$-th Fibonacci number $F(m)$, then computes its $n$-bit binary encoding, $F(m) = f_n\dots f_1f_0$, 
and provides the $i$th user with sugared espresso if and only if $f_{n-i} = 1$. 
Also Alice gets the last cup of her batch (i.e., the one corresponding to $f_{\lfloor\frac{n}{2}\rfloor}$), and Bob gets the very last one (i.e., for $f_0$). Atomic propositions $\prop{sugar_{Alice}}$ (resp.~$\prop{sugar_{Bob}}$) indicate whether Alice (resp.~Bob) get their espresso sweet. Both agents have perfect information about the state of the game.

We can formalize the scenario by a model template $\mclass[\coffee]$ with $\Agt=\set{Alice,Bob}$ and $Act=\set{request,skip}$.
The states are $\States = \set{i/j \mid i\le j, j\le n}$, where $j$ indicates the number of decisions already made, and $i$ is the count of requests for sugar.
The repertoires are:
$R_{Alice}(i/j) = \set{request,skip}$ if $j< \lceil\frac{n}{2}\rceil$ and $R_{Alice}(i/j) = \set{skip}$ otherwise;
$R_{Bob}(i/j) = \set{request,skip}$ if $\lceil\frac{n}{2}\rceil\le j <n$ and else $R_{Bob}(i/j) = \set{skip}$.
First, Alice executes a sequence of $\lceil\frac{n}{2}\rceil$ choices, and the state of the model records the count of $request$ actions. Then, Bob executes his sequence of $\lfloor\frac{n}{2}\rfloor$ choices, and the counting continues. At the end, coffee is distributed, with the valuation of atomic propositions fixed accordingly.
The instance of $\mclass[\coffee]$ for $n=3$ is depicted in Figure~\ref{fig:fibonacci}. The states controlled by Alice are set in grey; only the relevant actions are indicated in transition labels.
\end{example}

\begin{figure}[t]
\centering
\begin{tikzpicture}[->,>=stealth',shorten >=1pt,auto,node distance=2.1cm,transform shape,semithick,scale=0.9]
\tikzstyle{state}=[circle,fill=none,draw=black,text=black,minimum size=0.6cm]
\tikzstyle{alicestate}=[state,fill=lightgrey]
\tikzstyle{strat}=[very thick]

\node[alicestate,initial left] (s00) {$0/0$}; \node[alicestate] (s01) [below left=0.6cm and 0.6cm of s00] {$0/1$};
\node[alicestate] (s11) [below right=0.6cm and 0.6cm of s00] {$1/1$};
\node[state] (s02) [below left=0.6cm and 0.6cm of s01] {$0/2$};
\node[state] (s12) [below right=0.6cm and 0.6cm of s01] {$1/2$};
\node[state] (s22) [below right=0.6cm and 0.6cm of s11] {$2/2$};
\node[state] (s03) [below left=0.6cm and 0.6cm of s02] {$0/3$};
\node[state] (s13) [below left=0.6cm and 0.6cm of s12, label=below:{$\prop{sugar_{Bob}}$}] {$1/3$};
\node[state] (s23) [below right=0.6cm and 0.6cm of s12, label=below:{$\prop{sugar_{Bob}}$}] {$2/3$};
\node[state] (s33) [below right=0.6cm and 0.6cm of s22, label=below:{$\prop{sugar_{Alice}}$}] {$3/3$};

\path
(s00)
edge node[near start, left] {$skip$}	(s01)
  edge node[near start, right] {$request$} (s11)
(s01)
  edge node[near start, left] {$skip$}	(s02)
  edge node[midway,left=-15pt] {$request$}	(s12)
(s11)
  edge node[midway,right=-12pt]{$skip$} (s12)
  edge node[near start, right]{$request$} (s22)
(s02)
  edge node[near start, left] {$skip$}	(s03)
  edge node[midway,left=-15pt] {$request$}	(s13)
(s12)
  edge node[midway,right=-12pt]{$skip$} (s13)
  edge node[midway,left=-15pt]{$request$} (s23)
(s22)
  edge node[midway,right=-12pt]{$skip$} (s23)
  edge node[near start, right]{$request$} (s33)
(s03)
  edge [loop below] (s03)
(s13)
  edge [loop left] (s13)
(s23)
  edge [loop left] (s23)
(s33)
  edge [loop left] (s33)
;
 \end{tikzpicture}
\vspace*{-10pt}
\caption{Fibonacci coffee machine $\model^\coffee_3$ for $n=3$ }
\label{fig:fibonacci}
\vspace*{-10pt}
\end{figure}
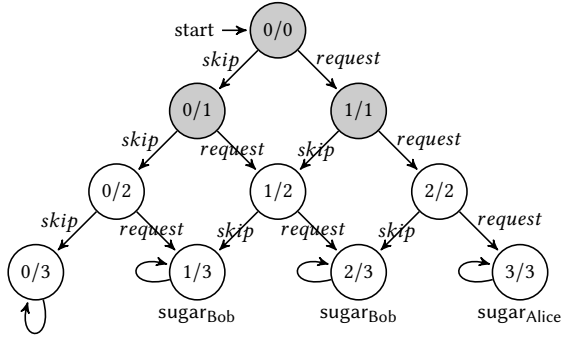

\WJ{Originally planned running examples: (1) CPA2 security (security against adaptive chosen-plaintext attacks for multiple encryptions); (2) solving labyrinth puzzles.
There are some problems with those examples, see the comments at the end of Section~\ref{sec:compgames}. }

\subsection{Computational Strategies}\label{sec:strats}

Strategies are represented by Turing machines that take a history of observations and produce a decision.
Strategy templates are Turing machines that take a model and a history, and produce a decision.

\begin{definition}[Computational strategy]\label{def:strats}
A \emph{computational strategy} $\strat$ for agent $a$ in model $\model$ is an input/output Turing machine with $1$ input tape and $1$ output tape, that takes as input a sequence of $a$'s observations, and returns as output an action from $a$'s repertoire in the model.
We require that the machine is deterministic and total, i.e., terminating on every input.

A \emph{general computational strategy} $\strtemp$ for agent $a$ in family $\mclass$ is an input/output Turing machine with 2 input tapes and 1 output tape, that takes as input the explicit representation of a model and a sequence of $a$'s observations, and returns as output an action from $a$'s repertoire in the model.

If we fix a model $M\in\mclass$ on the first input tape, the strategy template $\strtemp$ gets instantiated to the computational strategy $\strat = \strtemp(\model)$.
\end{definition}

\begin{example}\label{ex:fibonacci-strat}
The following general strategy of Bob can be easily implemented by a Turing machine of Definition~\ref{def:strats}.
Extract the scalability parameter $n$ from the model encoding on the first tape, and take the current state $i/j$ from the history encoding on the second tape.
If $j<n$, then write the encoding of $skip$ on the output tape. Else, compute $F(i)$ by means of the straightforward recursive algorithm, and write $skip$ if the least significant bit of the result is $1$, otherwise write $request$.
\end{example}

\begin{remark}[Representations of concurrent game structures]\label{rem:strats}
We are going to ask about the existence of winning strategies within a given complexity bound.
This might crucially depend on how models and observations are represented as inputs to the Turing machine that implements a strategy.
In the rest of this paper, we assume that the elements of $\Agt$, $Act$, and $\States$ are ordered in an arbitrary way, and identified by their indices.
The indices are encoded in binary, i.e., with appropriate sequences of symbols $\set{0,1}$.
Moreover, sets, tuples, and sequences can be represented by enumerating the elements, separated by the special symbol $\#$, with the opening and closing brackets encoded as $\#00\#$ and $\#01\#$, respectively.
This way, the representation of an iCGS $\model$ on the input tape incurs only a logarithmic increase with respect to the ``abstract'' size of $\model$, understood as the total number of states, transitions, and indistinguishability pairs in $\model$.
Similarly, the representation of a history of observations is at most logarithmically longer than the abstract length of the history.

Importantly, the Turing machines representing computational strategies use the ternary tape alphabet $\Gamma = \set{0,1,\#}$.
\end{remark}

\begin{definition}[Collective strategies of coalitions]\label{def:collective-strats}
A computational strategy (resp.~general computational strategy) for coalition $A\subseteq\Agt$ is simply a tuple of computational strategies (resp.~general computational strategies), one per agent $a\in A$.
We denote the set of computational strategies for $A$ in $\model$ by $\strset[\model]{A}$, and the set of general computational strategies for $A$ in $\mclass$ by $\genstrset[\mclass]{A}$.
\end{definition}

We use $out(\model,\strat)$ to denote the \emph{outcome set} of strategy $\strat$ in model $\model$, i.e., the set of infinite paths in $\model$, consistent with $\strat$ (and analogously for model templates).

\begin{definition}[Outcome]
	Given a computational strategy $s \in \strset[\model]{A}$, we define
$out(\model,\strat) =$  $\{ \lambda=\initstate,\state_1,\state_2\ldots \mid$ for each $i=0,1,\ldots$ there exists
      $\vec{\alpha} \in \repert_\Agt(\state_i)$ such that for each $a\in A$,
      $\vec{(\alpha^i)}_{a} = s_a([q_0]_{\sim_a} \cdot [q_1]_{\sim_a} \cdot \ldots \cdot [q_i]_{\sim_a})$
      and $q_{i+1} = t(q_{i},\vec{\alpha^i})$.

Moreover, the outcome of a general computational strategy $\strtemp \in \genstrset[\mclass]{A}$ is defined as $out(\mclass,\strtemp) = \bigcup_{\model\in\mclass} out(\model,\strtemp(\model))$.
\end{definition}

\begin{example}\label{ex:fibonacci-outcome}
The outcome of the strategy $\genstr_{Bob}$, described in Example~\ref{ex:fibonacci-strat}, in model $\model^\coffee_3$ of Figure~\ref{fig:fibonacci} is
$out(\model^\coffee_3,\genstr_{Bob}) = \set{ 0/0\cdot 0/1 \cdot 0/2\cdot (1/3)^\omega ,\ 0/0\cdot 1/1\cdot 1/2\cdot (1/3)^\omega}$.

The outcome of $\genstr_{Bob}$ in template $\mclass[\coffee]$ is 
$out(\mclass[\coffee],\genstr_{Bob}) = 
\set{ 0/0 \dots i/\lceil\frac{n}{2}\rceil \dots i/(n\!-\!1) \cdot i/n^\omega \mid n\in\Nat, 0\le i\le \lceil\frac{n}{2}\rceil} $ if $F(i)$ 
 is odd, and 
$\set{0/0 \dots i/\lceil\frac{n}{2}\rceil \dots  i/(n\!-\!1)\cdot (i\!+\!1)/n^\omega \mid n\in\Nat, 0\le i\le \lceil\frac{n}{2}\rceil}$ 
if \\ 
$F(i)$ is even. 
\end{example}

\subsection{Complexity of Models and Strategies}\label{sec:complexity}

The time complexity of a strategy is defined as the number of steps needed to produce decisions, measured with respect to the size of the input.

\begin{definition}[Time complexity of general strategies]\label{def:strat-complexity}
Consider a model $\model$ and a computational strategy $\strat_a$ of agent $a$ in $\model$.
Let $runs(\strat_a,\inseq)$ denote the set of runs of strategy $\strat_a$ on the input $\enc(\inseq)$, encoding a sequence of observations $\inseq$.
Note that $\strat_a$ is a Turing machine as per Definition~\ref{def:strats}, and its \emph{runs on $\enc(\inseq)$} should not be confused with its \emph{outcome paths in $\model$}. 
For a run $\run$, we use $|\run|$ to denote the number of steps on the run. Clearly, $|\run| \in \Nat$ if the run is terminating, and $|\run| = \infty$ otherwise.
Similarly, $|\enc(\inseq)|$ and $|\enc(\model)|$ denote the encoding size of the input sequence $\inseq$ and the model $\model$, respectively.

The (pessimistic) \emph{time complexity of general strategy $\strtemp$} is captured by function
$\timecompl{\strtemp_a}: \Nat\times\Nat \then \Nat$, defined as follows:
\begin{eqnarray*}
\timecompl{\strtemp_a}(n,i) &=& \sup\{|\run|\ \mid\ \run\in runs(\strtemp_a(\model),\inseq), \\
  && \qquad\qquad\model\in\mclass, |\enc(\model)|=n, |\enc(\inseq)|=i\}.
\end{eqnarray*}
\end{definition}

An alternative way to define the complexity of a computational strategy would be to set it against the \emph{abstract} sizes of the model and the history.
Note that, under the assumption put forward in Remark~\ref{rem:strats}, such complexity would be at most logarithmically higher than the one defined above.

\begin{definition}[Complexity of collective general strategies]\label{def:collective-complexity}
The \emph{complexity of collective general strategy $\strtemp_A$} is defined as\ $\timecompl{\strtemp_A}(n,i) = \max_{a\in A} \timecompl{\strtemp_a}(n,i)$.
\end{definition}

\begin{example}\label{ex:fibonacci-complexity}
The strategy $\genstr_{Bob}$ of Example~\ref{ex:fibonacci-strat} runs in $O(r^n)$ where $r$ is the \emph{golden ratio}.
This is due to the complexity of the straightforward (and naive) computation of Fibonacci numbers~\cite{Dasdan18fibonacci}, used in the strategy.
\end{example}

Note that the definitions in this section can be extended to deterministic \emph{space} complexity in a straightforward way.

\subsection{Ability in Computational Strategies}\label{sec:uniform}

The \emph{uniform} computational ability captures the existence of a computationally feasible strategy template that uniformly wins {all the games in $\mclass$}.

\begin{definition}[Uniform computational ability]\label{def:uniform-ability}
Let $\mclass$ be a model template, $\varphi$ an \LTL objective in $\mclass$, and $\complx$ a complexity class.
Agents $A\subseteq\Agt$ have \emph{uniform $\complx$-ability} in $\mclass$ for $\varphi$ (written: $\mclass,A \stratunif[\complx] \varphi$) if there exists a general strategy $\strtemp_A$ for $A$, such that:
\begin{enumerate2}
\item For every path $\lambda\in out(\mclass,\strtemp_A)$, we have that $\lambda \satisf[\LTL] \varphi$, and
\item There exists $f\in\complx$ such that $\forall n,i\ .\ \timecompl{\strtemp_A}(n,i) \le f(n,i)$.
\end{enumerate2}
\end{definition}

Definition~\ref{def:uniform-ability} can be easily generalized to $\omega$-regular objectives~\cite{McNaughton66omega-regular,Alfaro00omegareg-games}.

\begin{example}\label{ex:fibonacci-uniform}
Going back to our running example, we observe that even the grand coalition cannot uniformly provide both Alice {and} Bob with sugared espresso, regardless of the complexity bound. This is because Alice's bit of $F(n)$ will always be $0$ for large enough values of $n$.
In consequence,
$\mclass[\coffee],\set{Alice,Bob} \notstratunif[\complx] \Sometm\prop{sugar_{Alice}} \land \Sometm\prop{sugar_{Bob}}$, for all complexity classes $\complx$.

On the other hand, Bob can make sure that he eventually gets his espresso with sugar. For each pair $F(k), F(k+1)$ of subsequent numbers, at least one of them is odd and hence with $f_0=1$. Thus, Bob's strategy of Example~\ref{ex:fibonacci-strat} enforces $\Sometm\prop{sugar_{Bob}}$. By the observation in Example~\ref{ex:fibonacci-complexity}, we get that $\mclass[\coffee],\set{Bob} \stratunif[\EXPTIME] \Sometm\prop{sugar_{Bob}}$.

Notice further that the strategy can be improved by computing Fibonacci numbers using the recursive algorithm with memoization (i.e., storing the intermediate results), running in time $O(n)$.
Thus, we also get $\mclass[\coffee],\set{Bob} \stratunif[\PTIME] \Sometm\prop{sugar_{Bob}}$.
Finally, if the value of $n$ is encoded explicitly in the representation of $\model_m^\coffee$, then the strategy can be further optimized to $O(\log n)$ by means of an algorithm based on a matrix representation of the Fibonacci sequence and exponentiation by repeated squaring~\cite{Dasdan18fibonacci}.
In consequence, $\mclass[\coffee],\set{Bob} \stratunif[\LOGTIME] \Sometm\prop{sugar_{Bob}}$.

Observe also that, while Alice cannot make her espresso sweet, she can ensure that she gets it bitter. To this end, it suffices that she always executes $skip$, regardless of what she sees in the input.
Thus, $\mclass[\coffee],\set{Alice} \stratunif[O(1)] \Always\prop{sugar_{Alice}}$.
\end{example}

\begin{remark}
In Definition~\ref{def:uniform-ability}, we only require the existence of a winning computational strategy for agents $A$ within the given complexity class. We do \emph{not} require that the agent can also computationally \emph{verify} that the strategy is winning and appropriately bounded. This is somewhat analogous to the notions of \emph{objective} vs. \emph{subjective} strategic ability, see~\cite[Section~11.5.2]{Agotnes15handbook} for a discussion.
\end{remark}

\WJ{It might be worth formalizing the other notion in the future. BTW, I don't think it is completely trivial to formalize.}

\subsection{Hierarchy of Computational Abilities}\label{sec:hierarchy-uniform}

We will now show that the concept of ability, proposed in Definition~\ref{def:uniform-ability}, is nontrivial.
In particular, uniform ability in \Ptime does not subsume uniform ability in \EXPTIME (unless \Ptime=\NP).

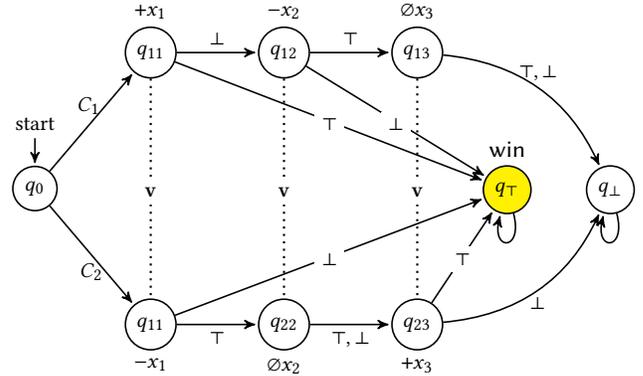
\begin{figure}[t]
\centering
\begin{tikzpicture}[->,>=stealth',shorten >=1pt,auto,node distance=2.1cm,transform shape,semithick,scale=0.9]
\tikzstyle{state}=[circle,fill=none,draw=black,text=black,minimum size=0.6cm]
\tikzstyle{initstate}=[state]
\tikzstyle{winstate}=[state,fill=yellow]
\tikzstyle{strat}=[very thick]
\tikzstyle{epist}=[-,dotted,thick]
\newcommand{\resc}{0.8}

\node[initstate,initial above] (q0) {$q_0$}; \node[state] (q11) [above right=1.5cm and \resc*1.5cm of q0, label=above:{$+x_1$}] {$q_{11}$};
\node[state] (q12) [right=\resc*1.5cm of q11, label=above:{$-x_2$}] {$q_{12}$};
\node[state] (q13) [right=\resc*1.5cm of q12, label=above:{$\varnothing x_3$}] {$q_{13}$};
\node[state] (q21) [below right=1.5cm and \resc*1.5cm of q0, label=below:{$-x_1$}] {$q_{11}$};
\node[state] (q22) [right=\resc*1.5cm of q21, label=below:{$\varnothing x_2$}] {$q_{22}$};
\node[state] (q23) [right=\resc*1.5cm of q22, label=below:{$+x_3$}] {$q_{23}$};
\node[winstate] (qwin) [below right=1.5cm and \resc*1cm of q13, label=above:{\large $\prop{win}$}] {$q_\true$};
\node[state] (qlose) [right=\resc*1cm of qwin] {$q_\false$};

\path
(q0)
  edge node[midway,above] {$C_1\ $}	(q11)
  edge node[midway,below] {$C_2$}	(q21)
(q11)
  edge node[midway,above] {$\false$} (q12)
  edge node[midway,above=-10pt] {\onlabel{$\true$}}	(qwin)
  edge[epist] node[midway,above=-10pt]{\onlabel{$\verf$}} (q21)
(q12)
  edge node[midway,above] {$\true$} (q13)
  edge node[midway,above=-10pt] {\onlabel{$\false$}}	(qwin)
  edge[epist] node[midway,above=-10pt]{\onlabel{$\verf$}} (q22)
(q13)
  edge[bend left] node[midway,above] {$\true,\false$}	(qlose)
  edge[epist] node[midway,above=-10pt]{\onlabel{$\verf$}} (q23)
(q21)
  edge node[midway,below] {$\true$} (q22)
  edge node[midway,above=-10pt] {\onlabel{$\false$}}	(qwin)
(q22)
  edge node[midway,below] {$\true,\false$} (q23)
(q23)
  edge node[midway,above=-10pt] {\onlabel{$\true$}}	(qwin)
  edge[bend right] node[midway,below] {$\false$}	(qlose)
(qwin)
  edge [loop below] (qwin)
(qlose)
  edge [loop below] (qlose);
 \end{tikzpicture}
\caption{The iCGS $M_\phi$ for $\phi \equiv (x_1 \lor \neg x_2) \land (\neg x_1 \lor x_3)$. The refuter $\reft$ controls the choice at the initial state $q_0$, and the verifier $\verf$ makes all the other choices. We only indicate the action selected by the active player for each transition. }
\label{fig:satgame}
\end{figure}

\begin{theorem}\label{prop:uniform-nontrivial}
There exists a model template $\mclass$, coalition $A$ in $\mclass$, and \LTL objective $\varphi$, such that $\mclass,A \stratunif[\EXPTIME] \varphi$ but not $\mclass,A \stratunif[\Ptime] \varphi$.
This applies even when restricting $A$ to singleton coalitions and $\varphi$ to reachability objectives.
\end{theorem}
\begin{proof}
The idea is to take a suitable game encoding of the Boolean satisfiability problem (SAT), and show that (1) the proponent has a general exponential-time strategy to win all the ``satisfiable'' games, but (2) the existence of a general polynomial-time strategy to do so would imply that SAT can be itself solved in deterministic polynomial time.
To this end, we take inspiration from the ``Boolean satisfiability games'' of~\cite[Section~3.1]{Schobbens04ATL} and~\cite[Section~3.1]{Jamroga06atlir-eumas}, which present two different reductions of the SAT problem to model checking of alternating time temporal logic with imperfect information and memoryless strategies. None of the reductions works for strategies with memory, but they can be combined and adapted to our case.

Our encoding of SAT works as follows.
Let $\phi$ be a Boolean formula in Conjunctive Normal Form with $n$ clauses and $k$ Boolean variables $x_1,\dots,x_k$. We assume w.l.o.g.~that the literals in each clause are ordered according to the underlying variable, and literals with the same variable have been reduced. Thus, each clause $C_i$ can be represented as $a_1^i x_1\land \dots \land a_k^i x_k$, with each $a_j^i\in\set{+,-,\varnothing}$ indicating whether $C_i$ includes a positive occurrence of $x_j$ ($+$), a negated occurrence of $x_j$ ($-$), or no occurrence of $x_j$ ($\varnothing$).

The corresponding iCGS $M_\phi$ is constructed as follows. There are two agents, the ``verifier'' $\verf$ and the ``refuter'' $\reft$. First, the refuter chooses a clause in $\phi$, without $\verf$ knowing which one has been chosen. Then, the verifier goes through the literals in the clause, one by one, declaring the value of the underlying proposition to be either $\true$ or $\false$. If it makes the literal true, the game proceeds to the winning state $q_\true$, otherwise the system moves to the next literal in the clause (literals with $\varnothing$ are unsatisfiable by definition). If this has been the last literal in the clause, then the game ends in the losing state $q_\false$. At each stage, the verifier only knows which literal in the clause (equivalently, which Boolean variable) is currently handled.
As an example, the satisfiability game for $\phi \equiv (x_1 \lor \neg x_2) \land (\neg x_1 \lor x_3)$ is presented in Figure~\ref{fig:satgame}.
Similarly to~\cite{Schobbens04ATL,Jamroga06atlir-eumas}, we have that
\begin{quote}
$\phi$ is satisfiable\quad iff\quad there exists an \iR strategy for $\verf$ in $M_\phi$ that surely obtains $\Sometm\prop{win}$. (*)
\end{quote}
Clearly, the winning \iR strategy can be always found by a brute-force search checking all \iR strategies in $M_\phi$ in exponential time w.r.t. the size of $M_\phi$. (**)

Take $\mclass[Sat]$ to be the set of satisfiability games $M_\phi$ for all the satisfiable formulas $\phi$,
let $\mclass[Unsat]$ be the set of satisfiability games for the unsatisfiable formulas,
and $\mclass[Bool]=\mclass[Sat]\cup\mclass[Unsat]$ the set of satisfiability games for all the Boolean formulas.
Clearly, each of those is a countable set of finite iCGS's; we assume the shortlex ordering of their elements when needed.
By (*) and (**), we get $\mclass[Sat],\set{\verf} \stratunif[\EXPTIME] \Sometm\prop{win}$.

Suppose that $\mclass[Sat],\set{\verf} \stratunif[\Ptime] \Sometm\prop{win}$ is also the case, i.e., the verifier has a polynomial-time general strategy $\genstr_\verf$ that obtains $\Sometm\prop{win}$ in $\mclass[Sat]$.
We will show that $\genstr_\verf$ can be used to construct a deterministic polynomial-time algorithm to solve Boolean satisfiability, thus contradicting \NP-completeness of the latter problem (unless $\Ptime=\NP$). The algorithm proceeds as follows:
\begin{enumerate2}
\item Given a Boolean formula $\phi$, construct its \extended{satisfiability}\short{sat.} game $M_\phi$;
\item Generate the prefixes of length $k+2$ for all the runs of $\genstr_\verf$ in $M_\phi$. If all \extended{those prefixes }end up in $q_\true$, return \emph{true}; otherwise, return \emph{false}.
\end{enumerate2}
Note that:
\begin{itemize2}
\item By assumption, $\genstr_\verf$ runs in polynomial time for any input over alphabet $\set{0,1,\#}$, in particular for the representation of any $M_\phi\in\mclass[Bool]$ stored on the first input tape.
    However, it is only guaranteed to implement a \emph{winning} strategy for inputs representing iCGS's in $\mclass[Sat]$. In fact, we know that such strategies for $\mclass[Unsat]$ cannot exist (and thus cannot be represented).
\item The outcome of any $\verf$'s strategy in $M_\phi$ contains exactly $n$ paths, one for each possible choice of the refuter in $q_0$.
\item Each path in $M_\phi$ reaches either $q_\true$ or $q_\false$ in exactly $k+2$ steps\extended{ (counting the initial state)}, and stays there forever from then on.
\end{itemize2}

We claim that the procedure runs in polynomial time, and returns \textit{true} iff $\phi$ is satisfiable, which obtains the contradiction and completes the proof.
\end{proof}

\WJ{TO DO LATER: show that the hierarchy collapses for safety/reachability goals in perfect info games. What about general \LTL goals in perfect info games?}

\WJ{Possible problems with the approach -- to be addressed later:
\begin{enumerate}
\item In order to analyze e.g. CPA security in our framework, we would need to include all the possible encryption keys in the model, and there are exponentially many of them wrt the length of the key. Then, the brute-force strategy of the attacker would be actually of linear time wrt the size of the model. Obviously, we don't want that. The problem is, we want the strategy of the attacker to be polynomial-time wrt an \emph{abstract view} of the procedure, with some parts given by oracles whose complexity is not a part of the model, cf. e.g. the definition of CPA security in (Katz and Lindell 2021).

\item Our complexity definition is in line with the standard one for Turing machines: the complexity is measured as a function of the whole input, i.e., both the model encoding and the history of observations. If the model has an ``idling'' loop for the concerned agent, it allows them to ``cheat'' around the complexity bound -- see the remark below. I am actually not sure if it's a bug or maybe a good feature. Possibly, it exposes the conceptual difference between constraining the complexity of the \emph{reasoning process} (captured by the number of steps of the TM implementing the strategy) vs. the number of \emph{game steps} in the model until reaching the goal. The latter can be e.g. captured by asymptotically constrained temporal operators, such as $\Sometm^\Ptime$.
    The two factors seem to be implicitly merged in security definitions, which is a potential conceptual gap that we might look closer at in the future.

    A simple scenario that exemplifies the difference between the two kinds of constraints is labyrinth puzzles. For the class of acyclic labyrinths, there is a winning strategy template that computes the decision in constant time (``if the space in front is open, then walk on, else turn left''). However, its execution requires linearly many \emph{game steps} until the winning state is achieved.

\item Another thing that we might want to include in our framework is the incremental nature of computing the next decision within a strategy. Normally, a program implementing the strategy would retain the computation outcome of the previous choice (for history $h_{i-1}$) and only take the next observation $o_i$ as the subsequent input to compute the decision for $h_{i-1}o_i$. If we want to go this way, we may consider \emph{persistent Turing machines} of (Goldin 2000) or \emph{interactive Turing machines} of (van Leeuwen and Widermann 2006) for our models of computational strategies.
\end{enumerate}

\begin{remark}
If the complexity bound is wrt to the length of the current history, the attacker can possibly cheat by first ``mindlessly'' producing an exponentially long history(*), and then breaking the encryption in polynomial time wrt the length of that history.

Re (*): The attacker computes the size of the model encryption $|enc(\model)|$ from the first input tape, then computes the length of the history on tape 2 and checks whether it is greater than $(1^{|enc(\model)|})_2$, which is a number greater or equal than $2^{|\model|}$. If false, the attacker idles; if true, proceeds to brute-force breaking of the encryption.
\end{remark}

}
 \section{Adaptive Computational Strategies}\label{sec:adaptive}

Uniform security notions are rather weak in the sense that the intruder must have a \emph{single} attack strategy that works for all values of the security parameter (e.g., all the possible key lengths).
If such a powerful attack strategy does not exist, the system is deemed secure.
In real life, we should be equally worried by the possibility that the intruder always has an efficient attack strategy against the \emph{current} value of the parameter. When we upgrade the security parameter to disable the attack, he takes some time to prepare a new attack strategy that efficiently works against the new security parameter, and so on.
We formalize this kind of strategic abilities here.
We note in passing the relation to non-uniform complexity theory~\cite{Cai12complexity}. However, unlike standard approaches that use Boolean circuits and measure complexity in terms of the circuit size or depth~\cite{Shannon49circuits,Furst81circuits,Cai12complexity}, we build directly on the time complexity of Turing machines.

\subsection{Non-Uniform Strategies and Abilities}

Complexity bounds are by definition asymptotic.
How can we define asymptotic bounds when considering a different strategy for each single model?
The solution is to look at \emph{collections} of such strategies, one per model.
In line with that, we generalize Definitions~\ref{def:strat-complexity} and~\ref{def:collective-complexity} to any mapping from a family of models to strategies in those models.

\begin{definition}[Strategy mappings]
$ST_A : \mclass \then \strset[\mclass]{A}$ is a \emph{strategy mapping for $A$} if $ST_A(\model)\in\strset[\model]{A}$ for every model $\model\in\mclass$. Intuitively $ST_A$ represents a family of computational strategies, each of them fitted for a different game that the agents in $A$ might come to play.
The outcome of a strategy mapping is defined as $out(\mclass,ST) = \bigcup_{\model\in\mclass} out(\model,ST(\model))$.
\end{definition}

\begin{definition}[Complexity of strategy mappings]
The time complexity of strategy mapping $ST_A$ is defined as
\begin{eqnarray*}
\timecompl{ST_a}(n,i) &=& \sup\{|\run|\ \mid\ \run\in runs(ST_a(\model),\inseq), \\
  && {}\qquad\qquad \model\in\mclass, |\model|=n, |\inseq|=i\} \\
\timecompl{ST_A}(n,i) &=& \max_{a\in A} \timecompl{ST_a}(n,i);
\end{eqnarray*}
where $|\model|$ is the abstract size of model $\model$, measured by the total number of states, transitions, and indistinguishability links in $\model$.
\end{definition}

\begin{definition}[Adaptive computational ability]
Let $\mclass$ be a family of models, $\varphi$ an \LTL objective in $\mclass$, and $\complx$ a complexity class.
Agents $A\subseteq\Agt$ have \emph{adaptive $\complx$-ability} in $\mclass$ for $\varphi$ (written: $\mclass,A \stratadapt[\complx] \varphi$) if there exists a strategy mapping $ST_A$ for $A$ in $\mclass$, such that:
\begin{enumerate}
\item For every path $\lambda\in out(\mclass,ST_A)$, we have that $\lambda \satisf[\LTL] \varphi$, and
\item There exists $f\in\complx$ such that $\forall n,i\ .\ \timecompl{ST_A}(n,i) \le f(n,i)$.
\end{enumerate}
\end{definition}

\WJ{\todo{revise to match the running example} For example, the coercer $c$ has adaptive polynomial-time ability to get candidate $1$ elected if he has a strategy mapping $ST_c$ that achieves $\Sometm\,\prop{win_1}$ on all outcome paths, such that the execution time of $ST_c$ is bounded by a polynomial.}

\WJ{To do later (hierarchy of adaptive abilities):
 (a) show that adaptive ability in \Ptime does not subsume adaptive ability in \EXPTIME,
 (b) similarly, the hierarchy collapses for safety/reachability goals in perfect info games;}

\subsection{Uniform vs. Adaptive Ability}

In this section we prove the main technical result of this paper, stating that uniform and adaptive abilities do not coincide.
In fact, uniform ability is strictly stronger than adaptive ability.

\begin{proposition}\label{prop:uniform2adaptive}
Uniform ability always implies adaptive ability.
Formally, for every model template $\mclass$, coalition $A$ in $\mclass$, complexity class $\complx$, and \LTL objective $\varphi$, if $\mclass,A \stratunif[\complx] \varphi$ then $\mclass,A \stratadapt[\complx] \varphi$.
\end{proposition}
\begin{proof}
Straightforward.
\end{proof}

\begin{theorem}\label{prop:adaptive2uniform}
Adaptive ability does not universally imply uniform ability.
Formally, there exists a model template $\mclass$, coalition $A$ in $\mclass$, complexity class $\complx$, and \LTL objective $\varphi$, such that $\mclass,A \stratadapt[\complx] \varphi$ but not $\mclass,A \stratunif[\complx] \varphi$.
This applies even when restricting $A$ to singleton coalitions and $\varphi$ to reachability objectives, and fixing $\complx=\Ptime$.
\end{theorem}
\begin{proof}
We take the game encoding $\mclass[Sat]$ of the SAT problem from the proof of Theorem~\ref{prop:uniform-nontrivial}.
It was already shown there that $\mclass[Sat],\set{\verf} \notstratunif[\Ptime] \Sometm\prop{win}$.
It remains to prove that $\verf$ has a collection of polynomial-time strategies to win each game in $\mclass[Sat]$.
Recall that
\begin{quote}
$\phi$ is satisfiable\quad iff\quad there exists an \iR strategy for $\verf$ in $M_\phi$ that surely obtains $\Sometm\prop{win}$. (*)
\end{quote}
Moreover, each \iR strategy of $\verf$ can be represented explicitly by a mapping from the number of observations in the history (which is bounded by $k+2$) to actions in $\set{\false,\true}$, thus also by a deterministic Turing machine with $k+2$ states, counting the number of observations on the input tape, and writing the appropriate decision as output. Clearly, the machine runs in time polynomial w.r.t.~$k$. (**)
By (*) and (**), we have that $\mclass[Sat],\set{\verf} \stratadapt[\Ptime] \Sometm\prop{win}$, i.e., the verifier has adaptive polynomial-time ability for $\Sometm\prop{win}$ in $\mclass[Sat]$.
\end{proof}

The proof of Theorem~\ref{prop:adaptive2uniform} relies on the assumption that $\PTIME\neq\NP$ and on the requirement that computational strategies must terminate \emph{on any input} (and not only the inputs that encode the game models in $\mclass$). We conjecture that the former assumption is essential, but the latter is not. Finding an alternative proof that does not employ the strong termination requirement is a matter for future work.
 \section{Model Checking for Computational Strategies}\label{sec:mcheck}

In this section, we formally define the problem of model checking for computationally bounded strategies, and study its decidability.
According to Rice's theorem, any non-trivial semantic property of Turing machines\footnote{
  A property is semantic if it depends only on the function computed by the machine. It is nontrivial if there are both machines that satisfy it, and ones that do not. }
is undecidable~\cite{Rice53decproblems}.
Since we restrict computation strategies to \emph{terminating} Turing machines, our framework is not directly applicable for Rice's theorem.
Still, one should expect mostly bad news here.

\subsection{Problem Definition}

First, we formally define what it means to verify the computational ability of a coalition in a given family of game models.
As before, we assume a finite nonempty set of all agents $\Agt$, a nonempty set of actions $Act$, and a countable set of atomic propositions $\Props$.

\begin{definition}[Model checking]\label{def:mcheck-uniform}
Model checking of uniform computational abilities is defined as the following decision problem.

\noindent
\textbf{Input:}
\begin{itemize2}
\item A specification of a countable family of iCGS's $\mclass = (\model_1,\model_2,\dots)$, given by a terminating Turing machine $gen$ with 1 input tape and 1 output tape; given $k\in\Nat$ as input, it generates a representation of the iCGS $\model_k$ as the output.

\item An \LTL formula $\varphi$, possibly using the encodings of atomic propositions in $\Props_{\mclass}$, defined by $gen$;
\item A coalition $A\subseteq\Agt$;
\item A complexity class $\complx$.
\end{itemize2}

\noindent
\textbf{Output:}
  \emph{true} if $\mclass,A \stratunif[\complx] \varphi$,
  otherwise \emph{false}.
\end{definition}

Model checking of adaptive computational abilities can be defined analogously, with the output being \emph{true} if $\mclass,A \stratadapt[\complx] \varphi$, and \emph{false} otherwise.

\subsection{Bad News: Undecidability}\label{sec:undecidable-singleagent}

The problem is undecidable even in the following special case.

\begin{theorem}\label{prop:undecidable-singleagent}
Model checking for uniform computational abilities is undecidable even for singleton coalitions with safety objectives and any complexity constraint from $O(n)$ upwards.
\end{theorem}
\begin{proof}
We prove the undecidability by a reduction of the non-halting problem for deterministic Turing machines with initially empty input tape.
Given a representation of such a Turing machine $TM$, we construct the model template $\mclass[TM]$ as the set of all the finite prefixes of the sole run of $TM$, with $\Agt={a}$, all the transitions labelled by the sole action label $idle$, and the accepting configurations of $TM$ labelled by the sole atomic proposition $\prop{accepting}$.
The template is Turing-representable in a straightforward way.
Note that there is only $1$ strategy available in $\mclass[TM]$, namely the automatic strategy that returns $idle$ regardless of the input. Moreover, the automatic strategy has complexity $O(1)$.

Then, $TM$ does not halt iff the automatic strategy enforces\linebreak $\Always\neg\prop{accepting}$ iff $\mclass[TM],\set{a} \stratunif[O(1)] \Always\neg\prop{accepting}$, which completes the reduction.
\end{proof}

\WJ{This should also work for adaptive abilities, right?}

\subsection{Further Bad News}\label{sec:undecidable-coalitions}

The undecidability result in Section~\ref{sec:undecidable-singleagent} relied on the fact that $\mclass$ was a countable family of models.
Here, we show that the problem can be undecidable even for a single game.
Note that the uniform and adaptive abilities coincide when $\mclass$ is a singleton, hence the proof below covers both variants of model checking.

\begin{theorem}\label{prop:undecidable-coalitions}
Model checking for computational abilities is undecidable even for coalitions of size 2 with safety objectives and polynomial-time complexity constraints over model templates that consist of a single game.
\end{theorem}
\begin{proof}
\newcommand{\GG}{\mathbb{G}}
\newcommand{\ecoal}[1]{\coop{#1}}

We recall here the construction from~\cite{Dima11undecidable}.
Namely, given a deterministic TM $M$ which starts with a blank tape, we build a 3-agent iCGS whose set of
atomic propositions is a singleton $\{ok\}$
such that $\ecoal{\{1,2\}} \GG \, ok$ iff $M$ halts.
We then show that the winning strategy for the coalition can be implemented by a polytime TM.

The informal description of the iCGS is as follows.
Initially, agent 3 (the adversary) chooses a tape cell $c_k$. Then, agents 1 and 2 must simulate the evolution of the tape cell $c_k$ on the unique run $\rho_M$ of $M$.
By the evolution of a tape cell up to some time instant, we mean the sequence consisting
of the alphabet symbol that the respective cell contains, plus an information bit about whether the R/W head
points to cell $c_k$ or not.
The 3rd agent may also choose the ``frontier'' between the two cells $c_k$ and $c_{k+1}$
in order to test that the two agents correctly simulate the exact moments when the R/W head moves from $c_k$ to $c_{k+1}$ or the other way round --
in the sense to be explained below.
In order to cope with a fixed number of choices, agent 3's choices are staggered among several steps, as in \cite{Dima11undecidable}.

The evolution of tape cell $c_k$ in $\rho_M$ up to time instant $n$ is encoded by a sequence of
$2n+4$ iCGS states which may be of two types: the tape symbol $a$ where $a$ is the contents of $c_k$ at some $i\le n$ configuration $\rho_M[i]$
and the R/W head points to another cell,
or the pair $(a,q)$ where $a$ is the contents of $c_k$ and the R/W head points to $c_k$ at $\rho_M[i]$.
Simulating the $(n+1)$-th TM transition (that is, $\rho_M[n+1]$) on such a sequence of states of length $2n+4$ (call it a history of length $2n+4$) 
increases the length of the sequence by $2$ new states, as explained in the following:

In a history of length $2n+3$ and ending in $a$, both agents should play $idle$ during the following 2 rounds
if $c_k$ is not modified by the $n$-th transition in $\rho_M$ and the R/W head is not moving to $c_k$.
On the contrary, when in the $n$-th transition of $\rho_M$, the R/W head "\textbf{moves right} and reaches $c_k$"
by applying some transition $(q,b,r,b',R)$, agent $2$ must play action $(q,r,R)$ while agent $1$ must play $idle$.
The state $(a,r)$ is appended to the history (which has now length $2n+4$).
Dually, when the R/W head "\textbf{moves left} and reaches $c_k$"
by applying some transition $(q,b,r,b',L)$, agent $1$ plays $(q,r,L)$ while agent $2$ plays $idle$, state $(a,r)$ is appended to the history.
Subsequently, if the $(n+1)$-th transition in $\rho_M$ is $(r,a,s,c,L)$, agent $2$ should choose action $(r,s,L)$, while agent $1$ should play $idle$.
The target state of this joint action is $c$, to account for the fact that $c_k$ holds a $c$. A dual situation occurs when the $(n+1)$-th transition in $\rho_M$ is
$(r,a,s,c,R)$, agent $1$ should choose action $(r,s,R)$, while agent $2$ should play $idle$, and the target state is $c$.

When the 3rd agent challenges the two agents with correctly simulating the evolution of the "frontier" between $c_k$ and $c_{k+1}$,
the evolution of this "frontier" in $\rho_M$ up to time instant $n$ is simulated by another type history of length $2n+4$,
composed by either a special state $tr$ (transitory) or some triples $(q,r,L)$ or $(q,r,R)$ with $q,r$ being states of the iCGS.
Their meaning is as follows:

When the history ends in $tr$ and the frontier is not crossed during the $n$-transition $\rho_M[n]$, 
the two agents should both play $idle$ during the following $2$ rounds and two states $tr$ are then appended to the history at positions $2n+5$ and $2n+6$.
When the frontier is crossed "\textbf{from right to left}", due to the application of some transition $(q,a,r,b,L)$,
agent $2$ should play $(q,r,L)$, while agent $1$ should play $idle$, and a new state $(q,r,L)$ is appended to the history.
Subsequently, agent $1$ should immediately issue the action $(q,r,L)$ while agent $2$ plays $idle$ and another $tr$ state is appended to the history.
The dual situation of crossing the frontier  "\textbf{from left to right}", due to the application of some transition $(q,a,r,b,R)$,
is simulated by agent $1$ playing $(q,r,R)$, while agent $2$ should play $idle$, and the state $(q,r,R)$ is appended in position $2n+5$.
Subsequently, in $(q,r,R)$, agent $2$ plays $(q,r,R)$ while agent $1$ plays $idle$ and state $tr$ is appended in position $2n+6$.

Indistinguishability is used to verify that the evolution of two adjacent tape cells
and their frontier is correctly simulated by the two protagonists.
Hence, agent $1$ does not distinguish a history 
which corresponds with the evolution of tape cell $c_k$ upto some configuration $n$ in the TM run
from a history of the same length which corresponds with the evolution of the frontier between $c_k$ and $c_{k+1}$ upto configuration $n$.
Similarly, agent $2$ does not distinguish a history 
which corresponds with the evolution of tape cell $c_{k+1}$ upto configuration $n$
from a history of the same length which corresponds with the evolution of the frontier between $c_k$ and $c_{k+1}$ upto configuration $n$.
Therefore, when agent $1$ plays action $(q,r,L)$ on a history corresponding with the evolution of tape cell $c_k$
upto instant $n$, she must play the same action on the history corresponding with the evolution of the frontier between
cells $c_{k-1}$ and $c_k$.

To enforce the fact that a "good" joint strategy must correspond with a nonblocking run,
all transitions with combinations of actions not listed above lead to an error state labeled $\neg ok$.
An example is a joint transition with label $((q,r,L), idle)$ from a state labeled $(a,s)$ with $s \neq q$ -- the "good" joint transition should be $((s,r,L),idle)$,
and only if $(s,a,r,b,L)$ is a legal TM transition.
As a consequence, when the TM $M$ halts, say, with the R/W head pointing $c_k$ which bears symbol $a$ and the current state being $s$,
all joint transitions from $(a,s)$ in the iCGS lead to the error state. So, in order to avoid $\neg ok$ states,
the protagonists must ensure that they correctly simulate the good transitions on each
"cell" or "frontier" history, and each such history can be unboundedly extended, that is, the TM $M$ has a nonblocking run.

We now show that, when the TM $M$ never halts, the unique joint strategy $\sigma$ described above can be implemented by
a polytime TM $T_\sigma$. Let us note first that the decisions that each agent must make are based on
the decomposition of the history in two parts:
\begin{enumerate}
\item A part of length $O(k)$ which identifies the
index of the simulated tape cell $c_k$, or the frontier between $c_k$ and and $c_{k+1}$.
\item A part of length $O(n)$ which identifies the length of $\rho_M$
at the end of which the two agents must decide how the tape cell $c_k$ or the frontier between $c_k$ and $c_{k+1}$.
\end{enumerate}

Then $T_\sigma$ must simulate, on an internal tape, the first $n$ transitions in the unique run $\rho_M$,
then read the contents of the $k$-th tape symbol plus the information whether the R/W head of $M$ crosses one of the
frontiers of $c_k$, and produce the appropriate decision for each agent.
During this simulation, only the first $n$ tape symbols need to be simulated, since the R/W head cannot visit more than $n$
cells during the first $n$ steps in $\rho_n$.
Each simulation of one transition of $M$ requires time polynomial in $n$, plus the generation of the $n$ blank tape cells of the initial configuration, requiring again time polynomial in $n$, plus collecting the information
for the $k$-th tape cell and/or its neighbor, requiring again time polynomial in $n$.
 \end{proof}

\subsection{In Quest for Decidable Fragments}

Theorems~\ref{prop:undecidable-singleagent} and~\ref{prop:undecidable-coalitions} show that the problem is inherently undecidable.
That was hardly unexpected. In fact, we are aware of no nontrivial decidability results for problems that answer if a particular subclass of instances of a more general decision problem can be solved in a given complexity class.
Even deciding whether an arbitrary subset of Boolean satisfiability (SAT) can be tackled in deterministic polynomial time has been elusive so far~\cite{Szeider08paramSAT}.

In this subsection, we list a couple of decidable cases for model checking of computational ability.
All of them follow from ``complexity saturation'' properties, i.e., if a winning strategy exists at all, it must be computable within the given complexity bound.
In this sense, the results are not extremely exciting.
We hope, however, that the insights can serve as a starting point to obtain more interesting characterizations in the future.

\begin{theorem}\label{prop:decidable-singleagent}
Model checking for uniform computational abilities in case of singleton coalitions $A = \{a\}$, singleton families of games $\mclass = \set{M_l}$, and complexity constraints from $O(n)$ up is decidable.
\end{theorem}
\begin{proof}
We transform the iCGS $M_l$ into a two-player game with perfect information and a parity condition
in which the states are labeled with observations of the agent $a$, such that a memoryless strategy for agent $a$ exists in $M_l$
if and only if the protagonist wins the two-player game.
For that, we first build a deterministic parity automaton $\mathcal{A}$ which encodes the LTL formula,
then compute the synchronous composition of $M_l$ and $\mathcal{A}$ to obtain
a two-player game with imperfect information and parity winning condition.
Then we utilize results from \cite{ChatterjeeD10,RaskinCDH07}
to transform this game into a parity game with perfect information with the desired property
and which the protagonist wins if and only if she wins with a memoryless strategy,
that is, a strategy in which the choice of an action depends only on the last state of the history.

Then the Turing machine for the general strategy $S_a$ for $a$ would have, as its set of states,
the states of the two-player game, and, for each sequence of observations,
would output the action prescribed by the winning strategy in the two-player game
at the unique state of that game occurring at the end of the sequence of observations.
Note that this Turing machine works in time linear in the length of the sequence of observations.
Note also that the size of the two-player game is at most quadruply-exponential in the size of the original game structure and the LTL formula,
but, since the size of the input model is treated as constant, this has no impact on the complexity of the winning strategy.
\end{proof}

For the next theorem we adapt the notion of multi-energy games
from \cite{ChatterjeeRR14,complexity-multi-energy-2015,pseudopoly-multi-energy-2015} in order to extract an interesting subproblem of model-checking.
An \emph{iCGS with $n$ counters} is an iCGS $\model$
in which all the components of $\model$ are the same as in Definition 2.1, but transitions are tuples of the form
$(q,\overline{\alpha},\gamma,\upsilon, r)$ where:
\begin{itemize}
	\item  $\overline{\alpha} \in Act^{\Agt}$, $\gamma$ is a \emph{positive} boolean combination of
	formulas of the type $c_i > 0$, where $c_i$ is one of a fixed set of $n$ counter variables.
	\item $\upsilon$ is a set of operations of the type $(\text{++}c_i)$ or $(\text{-\;-}c_i)$.
\end{itemize}
The ``semantics'' of the $n$-counter iCGS is an infinite iCGS whose states -- called "configurations" -- consist of a state and a tuple of $n$ counter values.
A transition labeled by $\overline{\alpha}$ exists between two configurations $K_1 = (q_1,c_1,\ldots, c_n)$ and $K_2 = (q_2,c_1',\ldots c_n')$ if
there is a transition $(q_1,\overline{\alpha},\gamma,\upsilon,q_2)$ such that $(c_1,\ldots,c_n) \models \gamma$ and $(c_1',\ldots,c_n')$ is
obtained from $(c_1,\ldots,c_n)$ by applying operations in $\upsilon$.

Given two integers $N_0$ and $k$ and a multi-energy model $M$, the \emph{$(N_0,k)$-bounded model} $M_{N_0,k}$ is
obtained by bounding the possible values of $c_1,\dots,c_n$ by $N_0+k$
in all configurations.
Additionally, an attempt to increment the value of some counter to more than $N_0+k$ results in no change of the counter value.

The class of \emph{energy-bounded families of iCGS} generated by a multi-energy model $M$ and an integer $N_0$
is the class $\mclass = (M_k)_{k\in \Nat}$ where $M_k = M_{N_0,k}$ is defined as above.

\begin{theorem}\label{prop:decidable-energy}
Model checking computational abilities in energy-bounded families of iCGS and singleton coalitions is decidable.
\end{theorem}
\vspace*{-5pt}
\begin{proof}
Note that, to check the existence of a strategy for agent $a$ which enforces an LTL formula for all $M_k$,
it suffices to check if there is a strategy that enforces the formula in the model $M_0$, since
each sequence of transitions which is feasible for $N_0$ is also feasible for a larger upper bound.
Furthermore, checking the existence of a strategy for $M_0$ requires
building the finite-state two-player game arena from $M_0$, and then utilizing again the classical results on constructing parity automata
from LTL formulas and \cite{ChatterjeeD10,RaskinCDH07}  for making parities visible and determinizing the resulting automaton.
As a result, we obtain a two-player finite-state parity game which can be won by the protagonist if and only if in the original iCGS $M_0$,
player $a$ would have a strategy to enforce the given LTL formula.
Again, this would require a memoryless strategy which provides a general strategy implementable by a Turing machine
working in linear time, as in the proof of Theorem~\ref{prop:decidable-singleagent}.
\end{proof}

\WJ{TO BE CONSIDERED IN THE FUTURE?
\begin{enumerate}
\item The number of states and actions in the models in $\mclass$ is bounded by a constant $m$ (fixed beforehand). Then, the set of games to consider is finite and bounded as well. \WJ{Still not immediate how that implies that verifying the existence of, say, \Ptime strategy is decidable.} \CD{It's not clear for me where's the infinite class of problem instances here. If $m$ is the maximal number of states in iCGS from any class $\mclass$, then we only have finitely many
instances of the problem and an algorithm always exists -- a big if-then-else which "knows exactly" whether the input instance has a strategy or not in the appropriate complexity class. This works regardless of the complexity class, but it's not really an interesting algorithmic problem.}
\WJ{I agree, let's discard it.}

\item Model checking of the "SlidingRobots" of Aminof and Rubin? Can we generalize their parameterized verification in a sensible way?
\end{enumerate}
} \section{Conclusions}\label{sec:conclusions}

We present the concept of computationally-bounded strategic ability,
which defines the agents'power in terms of their ability to synthesize a winning algorithm working within a given complexity bound,
which computes their choices in a parameterized game arena.
Our notion is inspired by cryptographic definitions of security, research on human-friendly strategies, and parameterized model-checking. 
We show that the hierarchy of computational abilities does not collapse,
since polynomially-bounded strategies are strictly weaker than exponentially-bounded ones.
Moreover, we show that the uniform and non-uniform variants of computational ability do not coincide.
We also define a class of model checking problems which can be thought of as a computationally parameterized variant of strategy synthesis for multi-agent systems with imperfect information. We show that the problem is undecidable even for very restricted classes of inputs,
and provide some simple cases in which the problem becomes decidable.

For future work, we plan to investigate in closer detail the connection with cryptographic definitions
of security, and to identify more relevant decidable classes of the model-checking problem,
in which the hierarchy of computational strategic abilities does not collapse above linear time.
The complexity of the model-checking problem for the decidable cases is another interesting path of future research.

\begin{acks}
The work has been supported by NCBR Poland and FNR Luxembourg under the PolLux/FNR-CORE projects STV (POLLUX-VII/1/2019 \& C18/IS/12685695/IS/STV/Ryan) and SpaceVote (POLLUX-XI/14/SpaceVote/2023).
\end{acks}

\bibliographystyle{ACM-Reference-Format}

\end{document}